\documentclass{amsart}

\pdfoutput=1

\usepackage{amsmath}
\usepackage{amssymb}
\usepackage{amsfonts}
\usepackage{graphicx}
\usepackage{texdraw}
\usepackage{graphpap}
\usepackage{enumitem}
\usepackage{color}
\usepackage[OT2,T1]{fontenc}
\usepackage{verbatim} 
\usepackage{mathtools}
\usepackage{multirow}
\usepackage{subcaption}
\usepackage[utf8]{inputenc}
\usepackage[english]{babel}
\usepackage[colorlinks=true, citecolor=blue, filecolor=black, linkcolor=black, urlcolor=black]{hyperref}
\usepackage{bbm}

\usepackage{mathtools}
\usepackage{todonotes}
\usepackage{appendix}

\theoremstyle{plain}
\newtheorem{thm}{Theorem}[section]

\newtheorem{lem}[thm]{Lemma}

\newtheorem{prop}[thm]{Proposition}

\newtheorem{rem}[thm]{Remark}

\theoremstyle{definition}

\numberwithin{equation}{section}

\newcommand{\Int}{\operatorname{Int}}

\newcommand{\R}{{\mathbb R}}

\newcommand{\eps}{{\varepsilon}}

\newcommand{\re}{\operatorname{Re}}
\newcommand{\im}{\operatorname{Im}}

\newcommand{\Ham}{\mathbf{H}}

\newcommand{\supp}{\operatorname{supp}}

\newcommand{\Lap}{\Delta}

\newcommand{\RN}[1]{%
	\textup{\uppercase\expandafter{\romannumeral#1}}%
}

\begin{document}

\author{Sung-Soo Byun}
\address{Center for Mathematical Challenges, Korea Institute for Advanced Study, 85 Hoegiro, Dongdaemun-gu, Seoul 02455, Republic of Korea}
\email{sungsoobyun@kias.re.kr}

\title[Equilibrium measure for the quadratic potentials with a point charge]{Planar equilibrium measure problem \\ in the quadratic fields with a point charge
}

\begin{abstract}
We consider a two-dimensional equilibrium measure problem under the presence of quadratic potentials with a point charge and derive the explicit shape of the associated droplets. 
This particularly shows that the topology of the droplets reveals a phase transition: (i) in the post-critical case, the droplets are doubly connected domain; (ii) in the critical case, they contain two merging type singular boundary points; (iii) in the pre-critical case, they consist of two disconnected components.  
From the random matrix theory point of view, our results provide the limiting spectral distribution of the complex and symplectic elliptic Ginibre ensembles conditioned to have zero eigenvalues, which can also be interpreted as a non-Hermitian extension of the Marchenko-Pastur law.   
\end{abstract}


\date{\today}

\thanks{Sung-Soo Byun was partially supported by Samsung Science and Technology Foundation (SSTF-BA1401-51) and by the National Research Foundation of Korea (NRF-2019R1A5A1028324) and by a KIAS Individual Grant (SP083201) via the Center for Mathematical Challenges at Korea Institute for Advanced Study.
}


\maketitle 

\section{Introduction and Main results}

In this paper, we study a planar equilibrium problem with logarithmic interaction under the influence of quadratic potentials with a point charge. 
This problem is purely deterministic, but its motivation comes from the random world, more precisely, from the random matrix theory or the theory of two-dimensional Coulomb gases in general.  
To be more concrete, for given points $\boldsymbol{\zeta}=(\zeta_j)_{j=1}^N \in \mathbb{C}^N$ of configurations, we consider the Hamiltonians 
\begin{align}
\Ham_N^{\mathbb{C}}(\boldsymbol{\zeta}) & := \sum_{ 1\le j<k \le N } \log \frac{1}{|\zeta_j-\zeta_k|^2} +N \sum_{j=1}^N W(\zeta_j), \label{Ham C}
\\
\Ham_N^{ \mathbb{H} } (\boldsymbol{\zeta}) & := \sum_{ 1\le j<k \le N } \log \frac{1}{|\zeta_j-\zeta_k|^2} + \sum_{ 1\le j \le k \le N } \log \frac{1}{ |\zeta_j-\bar{\zeta}_k|^2} +2N \sum_{j=1}^N W(\zeta_j). \label{Ham H}
\end{align}
Here $W: \mathbb{C}\to \R$ is a suitable function called the external potential. 
These are building blocks to define joint probability distributions
\begin{equation} \label{Gibbs}
d\mathbb{P}_{N,\beta}^\mathbb{C}(\boldsymbol{\zeta}) \propto e^{ -\frac{\beta}{2} \Ham_N^\mathbb{C}(\boldsymbol{\zeta}) } \prod_{j=1}^N \,dA(\zeta_j), 
\qquad d\mathbb{P}_{N,\beta}^{\mathbb{H}}(\boldsymbol{\zeta}) \propto e^{ -\frac{\beta}{2}\Ham_N^{\mathbb{H}}(\boldsymbol{\zeta}) } \prod_{j=1}^N \,dA(\zeta_j), 
\end{equation}
where $dA(\zeta)=d^2\zeta/\pi$ is the area measure and $\beta>0$ is the inverse temperature. 
Both point processes $\mathbb{P}_{N,\beta}^{ \mathbb{C} }$ and $\mathbb{P}_{N,\beta}^{\mathbb{H}}$ represent two-dimensional Coulomb gas ensembles \cite{forrester2010log,Lewin22,Serfaty}. 
In particular, if $\beta=2$, they are also called determinantal and Pfaffian Coulomb gases respectively due to their special integrable structures, see \cite{byun2022progress,byun2023progress} for recent reviews on this topic.
Furthermore, they have an interpretation as eigenvalues of non-Hermitian random matrices with unitary and symplectic symmetry. 
For instance, if $W(\zeta)=|\zeta|^2$, the ensembles \eqref{Gibbs} corresponds to the eigenvalues of complex and symplectic Ginibre matrices \cite{ginibre1965statistical}. 

One of the fundamental questions regarding such point processes is their macroscopic/global behaviours as $N\to \infty.$
For the case $\beta=2$, this can be regarded as a problem determining the limiting spectral distribution of given random matrices. 
The classical results in this direction include the circular law for the Ginibre ensembles. 
As expected from the structure of the Hamiltonians \eqref{Ham C} and \eqref{Ham H}, the macroscopic behaviours of the system can be effectively described using the logarithmic potential theory \cite{ST97}. 

For this purpose, let us briefly recap some basic notions in the logarithmic potential theory.
Given a compactly supported probability measure $\mu$ on $\mathbb{C}$, the weighted logarithmic energy $I_W[\mu]$ associated with the potential $W$ is given by 
\begin{equation} \label{energy}
I_W[\mu]:= \int_{ \mathbb{C}^2 } \log \frac{1}{ |z-w| }\, d\mu(z)\, d\mu(w) +\int_{ \mathbb{C} } W \,d\mu .
\end{equation}
For a general potential $W$ satisfying suitable conditions, there exists a unique probability measure $\mu_W$ which minimises $I_W[\mu]$.
Such a minimiser $\mu_W$ is called the \textbf{equilibrium measure} associated with $W$ and its support $S_W:= \supp (\mu_W)$ is called the \textbf{droplet}. 
Furthermore, if $W$ is $C^2$-smooth in a neighbourhood of $S_W$, it follows from Frostman's theorem that $\mu_W$ is absolutely continuous with respect to $dA$ and takes the form 
\begin{equation} \label{Frostman thm}
d\mu_W(z)= \Lap W(z) \cdot \mathbbm{1}_{ \{ z \in S_W \} } \,dA(z), 
\end{equation}
where $\Lap := \partial \bar{\partial}$ is the quarter of the usual Laplacian.

In relation with the point processes \eqref{Gibbs}, it is well known \cite{chafai2014first,MR2934715,HM13} that 
\begin{equation} \label{equilibrium convergence}
\mu_{N,W}:= \frac{1}{N}\sum_{j=1}^N \delta_{\zeta_j}  \to \mu_W
\end{equation}
in the weak star sense of measure. 
From the statistical physics viewpoint, this convergence is quite natural since the weighted energy $I_W$ in \eqref{energy} corresponds to the continuum limit of the discrete Hamiltonians \eqref{Ham C} and \eqref{Ham H} after proper renormalisations. 
(In the case of \eqref{Ham H}, it is required to further assume that $W(\zeta)=W(\bar{\zeta})$.)

Contrary to the density \eqref{Frostman thm} of the measure $\mu_W$, there is no general theory on the determination of its support $S_W$. 
(See however \cite{MR1097025} for a general theory on the  regularity and \cite{MR3454377} on the connectivity of the droplet associated with Hele-Shaw type potentials.)
This leads to the following natural question.
\smallskip 
\begin{center}
\textbf{For a given potential $W$, what is the precise shape of the associated droplet? }    
\end{center}
\smallskip 
In view of the energy functional \eqref{energy}, this is a typical form of an inverse problem in the potential theory and is called an equilibrium measure problem. 
Beyond the case when $W$ is radially symmetric (cf. \cite[Section IV.6]{ST97}), this problem is highly non-trivial even for some explicit potentials with a simple form, see \cite{MR4229527,MR3280250,MR3303173,balogh2015equilibrium, brauchart2018logarithmic,del2019equilibrium,legg2021logarithmic,criado2022vector} for some recent works.
Let us also stress that such a problem is important not only because it provides the intrinsic macroscopic behaviours of the point processes \eqref{Gibbs} but also because it plays the role of the first step to perform the Riemann-Hilbert analysis which gives rise to a more detailed statistical information ($k$-point functions) of the point processes, see \cite{MR3280250,MR3668632,MR3849128,MR2921180,MR3303173,MR2078083,MR3306308,MR3383811,MR3670735,lee2020strong,MR3939592} for extensive studies in this direction.
In this work, we aim to contribute to the equilibrium problems associated with the potentials \eqref{Q main} and \eqref{Q hat} below, which are of particular interest in the context of non-Hermitian random matrix theory.

\subsection{Main results}
For given parameters $\tau \in [0,1)$ and $c \ge 0$, we consider the potential
\begin{equation} \label{Q main}
Q(\zeta):=\frac{1}{1-\tau^2}\Big( |\zeta|^2-\tau \re \zeta^2 \Big)-2c\log|\zeta|. 
\end{equation}
When $\beta=2$, the ensembles \eqref{Gibbs} associated with $Q$ correspond to the distribution of random eigenvalues of the elliptic Ginibre matrices of size $(c+1)N$ conditioned to have zero eigenvalues with multiplicity $cN$. 
We mention that such a model with $c>0$ was also studied in the context of Quantum Chromodynamics \cite{akemann2001microscopic}. 

In \eqref{Q main}, the logarithmic term can be interpreted as an insertion of a point charge, see \cite{akemann2021scaling,charlier2022asymptotics,byun2022characteristic,byun2022almost,byun2022spherical} for recent investigations of the models \eqref{Gibbs} in this situation. 
Such insertion of a point charge has also been studied in the theory of planar orthogonal polynomials \cite{MR3280250,MR3668632,MR3849128,MR3670735,MR3962350,lee2020strong,berezin2022planar}. 
On the other hand, the parameter $\tau \in [0,1)$ captures the non-Hermiticity of the model. 
To be more precise, the models \eqref{Gibbs} associated with $Q$ interpolate the complex/symplectic Ginibre ensembles ($\tau=0$) with the Gaussian Unitary/Symplectic ensembles ($\tau=1$) conditioned to have zero eigenvalues, see Remark~\ref{Rem_MP} for further discussion in relation to our main results.  

For the case $c=0$, the terminology ``elliptic'' comes from the fact that the limiting spectrum $S_{\tau,0}$ is given by the ellipse 
\begin{equation}  \label{S tau 0} 
S_{\tau,0}:=\Big\{ (x,y) \in \R^2 : \Big( \frac{x}{1+\tau} \Big)^2+ \Big( \frac{y}{1-\tau} \Big)^2 \le 1  \Big\},
\end{equation}
which is known as the elliptic law, see e.g. \cite{choquard1987dielectric,forrester1996two}.
We refer to \cite{lee2016fine,akemann2022elliptic,ameur2021almost,molag2022edge,bothner2022complexEdge,bothner2022complexBulk} and references therein for more about the recent progress on the complex elliptic Ginibre ensembles and \cite{MR1928853,akemann2022skew,byun2021universal,byun2021wronskian} for their symplectic counterparts. 
For the rotationally invariant case when $\tau=0,$ it is easy to show that the associated droplet $S_{0,c}$ is given by 
\begin{equation}  \label{S 0 c}
S_{0,c}:=\Big\{ (x,y) \in \R^2 :  c \le x^2+y^2 \le 1+c  \Big\},
\end{equation}
see e.g. \cite[Section \RN{4}.6]{ST97} and \cite[Section 5.2]{byun2022progress}.

\bigskip 

The primary goal of this work is to determine the precise shape of the droplet associated with the potential \eqref{Q main} for general $\tau \in [0,1)$ and $c \ge 0$. 
For this, we set some notations. 
Let us write
\begin{equation} \label{tau critical}
\tau_c:=\frac{1}{1+2c}
\end{equation}
for the critical non-Hermiticity parameter. 
For $\tau \in  (\tau_c,1)$, we define
\begin{equation} \label{f rational main}
f(z)\equiv f_\tau(z) :=  \frac{(1+\tau)(1+2c)}{2}  \frac{(1-az)(z-a\tau)^2}{z(z-a)} , \qquad a=-\frac{1}{ \sqrt{\tau(1+2c)} }. 
\end{equation}
We are now ready to present our main result. 

\begin{thm} \label{Thm_insertion at 0 symmetric}
Let $Q$ be given by \eqref{Q main}. Then the droplet $S \equiv S_{\tau,c}=S_Q$ of the equilibrium measure 
\begin{equation} \label{eq msr Q}
d\mu_Q(z)= \frac{1}{1-\tau^2} \mathbbm{1}_S(z) \,dA(z) 
\end{equation}
is given as follows. 
\begin{itemize}
    \item \textup{\textbf{(Post-critical case)}} If $\tau \in (0, \tau_c]$, we have
\begin{equation} \label{S tau c}
S_{\tau,c}=\Big\{ (x,y) \in \R^2 : \Big( \frac{x}{(1+\tau)\sqrt{1+c}} \Big)^2+ \Big( \frac{y}{(1-\tau)\sqrt{1+c}} \Big)^2 \le 1 \, , \, \frac{x^2+y^2}{(1-\tau^2) c} \ge 1 \Big\}.
\end{equation} 
    \item \textup{\textbf{(Pre-critical case)}} If $\tau \in [\tau_c, 1)$, the droplet $S_{\tau,c}$ is the closure of the interior of the real analytic Jordan curves given by the image of the unit circle with respect to the map $z \mapsto \pm \sqrt{f(z)}$, where $f$ is given by \eqref{f rational main}. 
\end{itemize}
\end{thm}

\begin{figure}[h!]
	\begin{subfigure}{0.32\textwidth}
		\begin{center}	
			\includegraphics[width=\textwidth]{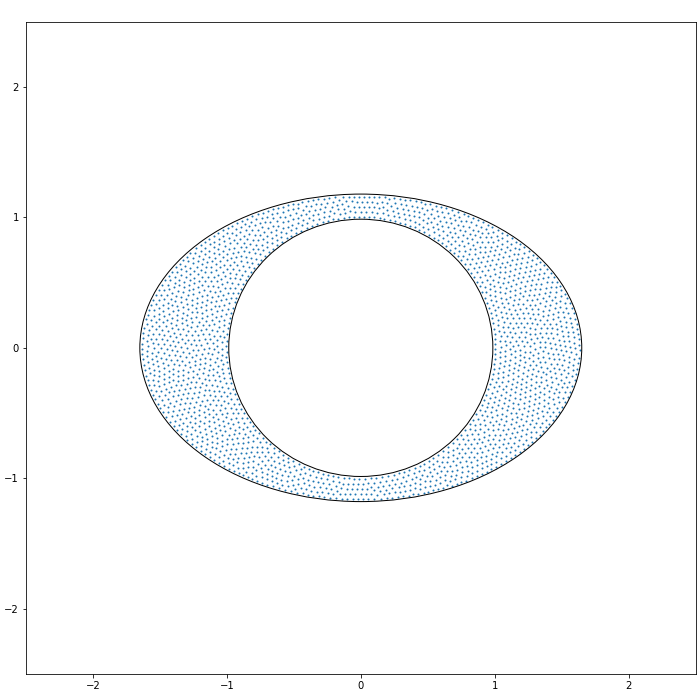}
		\end{center}
		\subcaption{$\tau=1/6<\tau_c$}
	\end{subfigure}	
	\begin{subfigure}{0.32\textwidth}
		\begin{center}	
			\includegraphics[width=\textwidth]{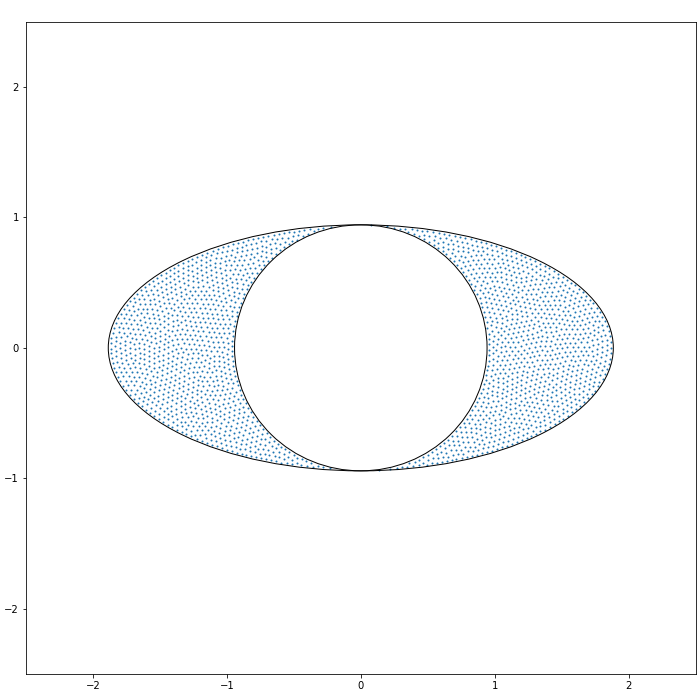}
		\end{center}
		\subcaption{$\tau=1/3=\tau_c$}
	\end{subfigure}	
	\begin{subfigure}{0.32\textwidth}
		\begin{center}	
			\includegraphics[width=\textwidth]{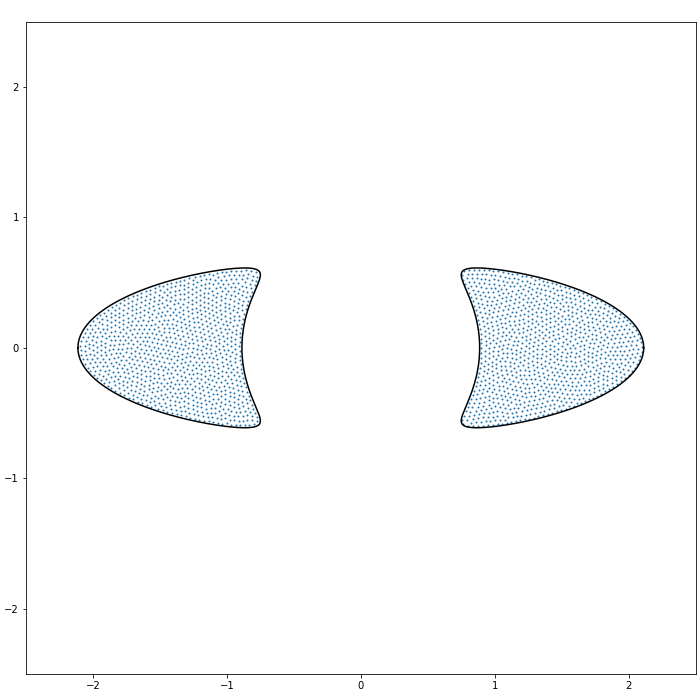}
		\end{center}
		\subcaption{$\tau=1/2>\tau_c$}
	\end{subfigure}	
	\caption{ The droplet $S_{\tau,c}$, where $c=1$ and a Fekete point configuration with $N=2048$. } \label{Fig_droplet}
\end{figure}

Note that if $c=0$ (resp., $\tau=0$), the droplet \eqref{S tau c} corresponds to \eqref{S tau 0} (resp., \eqref{S 0 c}). 
We mention that the post-critical case of Theorem~\ref{Thm_insertion at 0 symmetric} is indeed shown in a more general setup, see \eqref{Q p general} and Proposition~\ref{Thm_prec} below.

\begin{rem}[Phase transition of the droplet]
In Theorem~\ref{Thm_insertion at 0 symmetric}, we observe that if $c>0$, the topology of the droplet reveals a phase transition.
Namely, for the post-critical case when $\tau< \tau_c$, the droplet is a doubly connected domain, whereas for the pre-critical case $\tau > \tau_c$, it consists of two disconnected components. 
At criticality when $\tau=\tau_c$, the droplet contains two symmetric double points. 
We refer to \cite{MR3280250,balogh2015equilibrium, MR4229527,criado2022vector} for further models whose droplets reveal various phase transitions.
Let us also mention that recently, there have been several works on the models \eqref{Gibbs} with multi-component droplets, see e.g. \cite{MR3668632,byun2022determinantal,ameur2022two,ameur2022disk}.  
In this pre-critical regime, some theta-function oscillations are expected to appear for various kinds of statistics; cf. \cite{charlier2021large,ameur2022two}. 
The precise asymptotic behaviours of the partition function would also be interesting in connection with the conjecture that these depend on the Euler index of the droplets, see \cite{jancovici1994coulomb,byun2022partition} and \cite[Sections 4.1 and 5.3]{byun2022progress} for further discussion. 
\end{rem}

\begin{rem}[Fekete points and numerics]
A configuration $\{\zeta_j\}_{j=1}^N$ which makes the Hamiltonians \eqref{Ham C} or \eqref{Ham H} minimal is known as a Fekete configuration. 
This can be interpreted as the ensembles \eqref{Gibbs} with low temperature limit $\beta=\infty$, see e.g. \cite{sandier2012ginzburg,nodari2015renormalized,Ameur18low,ameur2022planar} and references therein. 
Since the droplet is independent of the inverse temperature $\beta>0$ (excluding the high-temperature regime \cite{MR3962973} when $\beta=O(1/N)$), the Fekete points are useful to numerically observe the shape of the droplets.  
In Figures~\ref{Fig_droplet} and \ref{Fig_droplet NW}, Fekete configurations associated with the Hamiltonian \eqref{Ham C} are also presented, which show good fits with Theorems~\ref{Thm_insertion at 0 symmetric} and \ref{Thm_insertion at 0}.  
\end{rem}

Notice that the potential \eqref{Q main} and the droplet $S_{\tau,c}$ are invariant under the map $\zeta \mapsto -\zeta$. 
We now discuss an equivalent formulation of Theorem~\ref{Thm_insertion at 0 symmetric} under the removal of such symmetry.
(See \cite[Section 1.3]{criado2022vector} for a similar discussion in a vector equilibrium problem on a sphere with point charges.)
The motivation for this formulation will be clear in the next subsection. 

For this purpose, we denote 
\begin{equation} \label{Q hat}
\widehat{Q}(\zeta):=\frac{2}{1-\tau^2}\Big( |\zeta|-\tau \re \zeta \Big)-2c\log|\zeta|.
\end{equation}
By definition, the potentials $Q$ in \eqref{Q main} and $\widehat{Q}$ in \eqref{Q hat} are related as 
\begin{equation}
Q(\zeta)=\frac12 \widehat{Q}(\zeta^2).
\end{equation}
Denoting by $\widehat{S}$ the droplet associated with $\widehat{Q}$, it follows from \cite[Lemma 1]{balogh2015equilibrium} that
\begin{equation}  \label{droplet relations}
S= \{ \zeta \in \mathbb{C}: \zeta^2 \in \widehat{S} \}.  
\end{equation}
Due to the relation \eqref{droplet relations} and 
\begin{equation}
\Delta \widehat{Q}(\zeta)=\frac{1}{2(1-\tau^2)} \frac{1}{|\zeta|},
\end{equation}
we have the following equivalent formulation of Theorem~\ref{Thm_insertion at 0 symmetric}. 

\begin{thm} \label{Thm_insertion at 0}
Let $\widehat{Q}$ be given by \eqref{Q hat}. Then the droplet $\widehat{S} \equiv \widehat{S}_{\tau,c}=S_{\widehat{Q}}$ of the equilibrium measure 
\begin{equation}   \label{eq msr Q hat}
d\mu_{\widehat{Q}}(\zeta)= \frac{1}{2(1-\tau^2)} \frac{1}{|\zeta|} \mathbbm{1}_{ \widehat{S} }(\zeta) \,dA(\zeta) 
\end{equation}
is given as follows. 
\begin{enumerate}[label=(\roman*)]
    \item \textup{\textbf{(Post-critical case)}} If $\tau \in (0, \tau_c]$, we have
\begin{equation}
\widehat{S}_{\tau,c}=\Big\{ (x,y) \in \R^2 : \Big( \frac{x-2\tau(1+c)}{(1+\tau^2)(1+c)} \Big)^2+ \Big( \frac{y}{(1-\tau^2)(1+c)} \Big)^2 \le 1 \, , \,  \frac{x^2+y^2}{ (1-\tau^2)^2c^2 } \ge 1 \Big\}.
\end{equation} 
    \item \textup{\textbf{(Pre-critical case)}} If $\tau \in [\tau_c, 1)$, the droplet $\widehat{S}_{\tau,c}$ is the closure of the interior of the real analytic Jordan curve given by the image of the unit circle with respect to the rational map $z \mapsto f(z)$, where $f$ is given by \eqref{f rational main}.  
\end{enumerate}
\end{thm}

\begin{figure}[h!]
	\begin{subfigure}{0.32\textwidth}
		\begin{center}	
			\includegraphics[width=\textwidth]{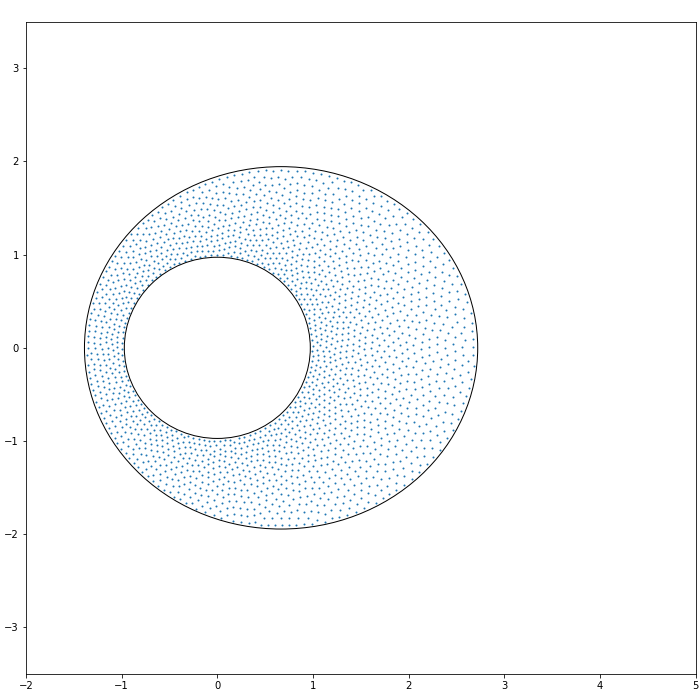}
		\end{center}
		\subcaption{$\tau=1/6<\tau_c$}
	\end{subfigure}	
	\begin{subfigure}{0.32\textwidth}
		\begin{center}	
			\includegraphics[width=\textwidth]{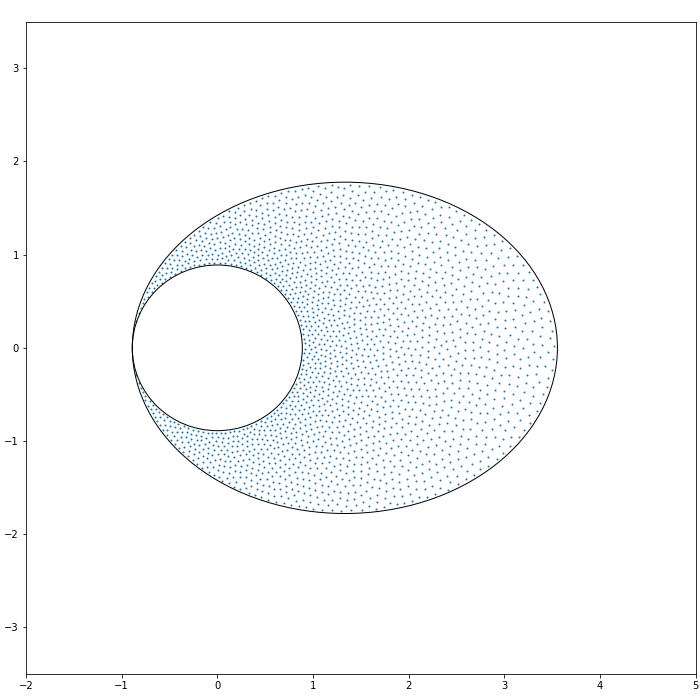}
		\end{center}
		\subcaption{$\tau=1/3=\tau_c$}
	\end{subfigure}	
	\begin{subfigure}{0.32\textwidth}
		\begin{center}	
			\includegraphics[width=\textwidth]{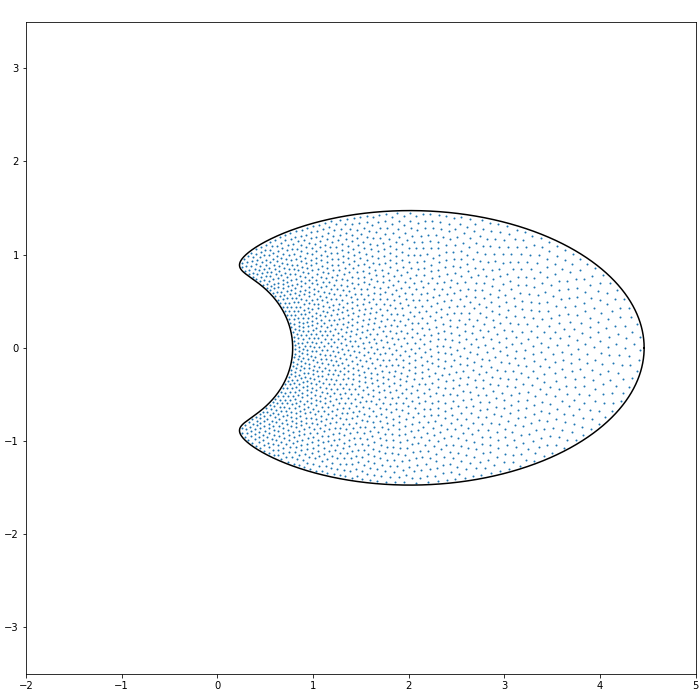}
		\end{center}
		\subcaption{$\tau=1/2>\tau_c$}
	\end{subfigure}	
	\caption{ The droplet $\widehat{S}_{\tau,c}$, where $c=1$ and a Fekete point configuration with $N=2048$.  } \label{Fig_droplet NW}
\end{figure}

\begin{rem}[Joukowsky transform in the critical case]
If $\tau=\tau_c$ with \eqref{tau critical}, we have $a=1/a=-1$.
Thus in this critical case, the rational function $f_\tau$ in \eqref{f rational main} is simplified as
\begin{equation} \label{f rational critical}
f_{\tau_c}(z)=(1+c) \frac{(z+\tau)^2}{z}= (1+c) \Big( z+ 2\tau + \frac{\tau^2}{z} \Big).
\end{equation}
Note that compared to the general case \eqref{f rational main}, there is one less zero and one less pole in \eqref{f rational critical}.  
Indeed, in the critical case, the rational map $f_{\tau_c}$ is a (shifted) Joukowsky transform
\begin{equation}
f_{\tau_c}: \mathbb{D}^c \to  \Big\{ (x,y) \in \R^2 : \Big( \frac{x-2\tau(1+c)}{(1+\tau^2)(1+c)} \Big)^2+ \Big( \frac{y}{(1-\tau^2)(1+c)} \Big)^2 \ge  1  \Big\}.
\end{equation}
In  \cite{MR4229527}, a similar type of Joukowsky transform was used to solve an equilibrium measure problem. 
For the models under consideration in the present work, due to a more complicated form of the rational function \eqref{f rational main}, the required analysis for the associated equilibrium problem turns out to be more involved.
\end{rem}

\begin{rem}[A non-Hermitian extension of the Marchenko-Pastur distribution] \label{Rem_MP}
In the Hermitian limit $\tau \uparrow 1,$ by \eqref{Q main} and \eqref{Q hat}, we have 
\begin{align}
&\lim_{ \tau \uparrow 1 } Q(x+iy)=V(x):=\begin{cases}
\displaystyle \frac{x^2}{2}-2c \log|x|, &\text{if }y=0,
\smallskip 
\\
+\infty &\text{otherwise}, 
\end{cases}
\\
& \lim_{ \tau \uparrow 1 } \widehat{Q}(x+iy)=\widehat{V}(x):=\begin{cases}
\displaystyle x-2c \log|x|, &\text{if }y=0, \, x>0,
\smallskip 
\\
+\infty &\text{otherwise}. 
\end{cases} 
\end{align}
Then the associated equilibrium measures are given by the well-known Marchenko-Pastur law (with squared variables) \cite[Proposition 3.4.1]{forrester2010log}, i.e. 
\begin{align}
d\mu_{ V }(x) & =  \frac{1 }{2\pi|x|}\,\sqrt{ (\lambda_+^2-x^2)(x^2-\lambda_-^2) }\cdot \mathbbm{1}_{ [-\lambda_+,-\lambda_- ] \cup [\lambda_-, \lambda_+ ] } \,dx ,  \label{MP square}
\\
d\mu_{ \widehat{V} }(x) & =  \frac{1 }{ 2\pi x}\,\sqrt{ (\lambda_+^2-x)(x-\lambda_-^2) }\cdot \mathbbm{1}_{ [\lambda_-^2, \lambda_+^2 ] } \,dx , \label{MP}
\end{align}
where $\lambda_{\pm}:= \sqrt{2c+1} \pm 1 $, cf. Remark~\ref{Rem_eq msr 1D}.	
Therefore one can interpret Theorem~\ref{Thm_insertion at 0 symmetric} (resp., Theorem~\ref{Thm_insertion at 0}) as a non-Hermitian generalisation of the Marchenko-Pastur distribution \eqref{MP square} (resp., \eqref{MP}), see \cite[Section 2]{ameur2021almost} for more about the geometric meaning with the notion of the statistical cross-section.  
We also refer to \cite{MR4229527} for another non-Hermitian extension of \eqref{MP square} and \eqref{MP} in the context of the chiral Ginibre ensembles. 
\end{rem}

\begin{rem}[Inclusion relations of the droplets]
Let us write
 \begin{align}
 S_1= \Big\{ (x,y)\in \R^2 : \Big( \frac{x}{1+\tau} \Big)^2+ \Big( \frac{y}{1-\tau} \Big)^2 \le 1+c\Big\}, \qquad    S_2:=\Big\{ (x,y)\in \R^2 : x^2+y^2 \le (1-\tau^2)c \Big\}
 \end{align}
 and denote by $\widehat{S}_{j}$ $(j=1,2)$ the image of $S_j$ under the map $z \mapsto z^2$. 
Then it follows from the definition \eqref{tau critical} that 
\begin{equation}
\tau \in (0,\tau_c) \qquad \textup{if and only if} \qquad S_1^c \cap S_2=\emptyset. 
\end{equation}

By Theorems~\ref{Thm_insertion at 0 symmetric} and \ref{Thm_insertion at 0}, for general $\tau \in [0,1)$ and $c \ge 0$, one can observe that 
\begin{equation} \label{inclusion principle}
S_{\tau,c} \subseteq S_1 \cap (\textup{Int} \, S_2)^c, \qquad  \widehat{S}_{\tau,c} \subseteq  \widehat{S}_1 \cap (\textup{Int}\, \widehat{S}_2)^c.
\end{equation}
Here, equality in \eqref{inclusion principle} holds if and only if in the post-critical case. 
(This property holds in a more general setup, see Proposition~\ref{Thm_prec}.)
On the other hand, in the pre-critical case one can interpret that the particles in $S_1^c \cap S_2$ are smeared out to $S_1 \cap S_2^c$ which makes the inclusion relations \eqref{inclusion principle} strictly hold, see Figure~\ref{Fig_droplet inclusion}.

\begin{figure}[h!]
	\begin{subfigure}{0.48\textwidth}
		\begin{center}	
			\includegraphics[width=\textwidth]{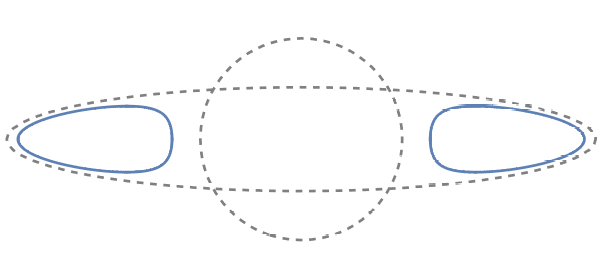}
		\end{center}
		\subcaption{$S_{\tau,c}$}
	\end{subfigure}	
	\begin{subfigure}{0.48\textwidth}
		\begin{center}	
			\includegraphics[width=\textwidth]{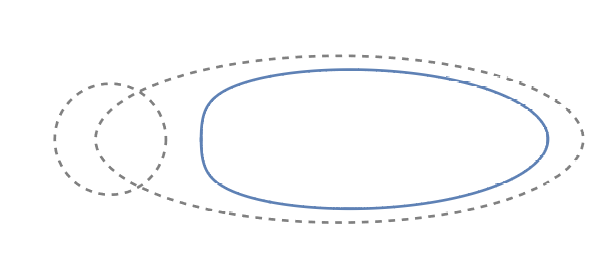}
		\end{center}
		\subcaption{$\widehat{S}_{\tau,c}$}
	\end{subfigure}	
	\caption{ The droplets $S_{\tau,c}$ and $\widehat{S}_{\tau,c}$ in the pre-critical case, where $c=2$ and $\tau=0.7>\tau_c$. Here, the dashed lines display the boundaries of $S_j$ and $\widehat{S}_j$ ($j=1,2$). }  \label{Fig_droplet inclusion}
\end{figure}

\end{rem}

\subsection{Outline of the proof}

Recall that $\mu_W$ is a unique minimiser of the energy \eqref{energy}. 
It is well known that the equilibrium measure $\mu_W$ is characterised by the variational conditions (see \cite[p.27]{ST97})
\begin{align}
\label{vari eq_prec}
\int \log \frac{1}{|\zeta-z|^2}\,d\mu_W(z)+W(\zeta)=C, \quad \text{q.e.} \quad \text{if }\zeta \in S_W; 
\smallskip 
\\
\int \log \frac{1}{|\zeta-z|^2}\,d\mu_W(z)+W(\zeta)\ge C, \quad \text{q.e.} \quad \text{if }\zeta \notin S_W. \label{vari ineq_prec}
\end{align} 
Here, q.e. stands for quasi-everywhere. 
(Nevertheless, this notion is not important in the sequel as we will show that for the models we consider the conditions \eqref{vari eq_prec} and \eqref{vari ineq_prec} indeed hold everywhere.)

Due to the uniqueness of the equilibrium measure, all we need to show is that if $W=Q$, then $\mu_Q$ in \eqref{eq msr Q} satisfies the variational principles \eqref{vari eq_prec} and \eqref{vari ineq_prec}. 
Equivalently, by \eqref{droplet relations}, it also suffices to show the variational principles for the equilibrium measure $\mu_{ \widehat{Q} }$ in \eqref{eq msr Q hat}. 

However, it is far from being obvious to obtain the ``correct candidate'' of the droplets.
Perhaps one may think that at least for the post-critical case, the shape of the droplet \eqref{S tau c} is quite natural given the well-known cases \eqref{S tau 0} and \eqref{S 0 c} as well as the fact that the area of $S_{\tau,c}$ should be $(1-\tau^2)\pi$.   
On the other hand, for the pre-critical case, one can easily notice that there is some secret behind deriving the explicit formula of the rational function \eqref{f rational main}. 
To derive the correct candidate, we use the conformal mapping method with the help of the Schwarz function, see Appendix~\ref{Appendix conformal mapping}.

\begin{rem}[Removal of symmetry]
We emphasise that the conformal mapping method does not work for the multi-component droplet, i.e. the pre-critical case of Theorem~\ref{Thm_insertion at 0 symmetric}. 
This is essentially due to the lack of the Riemann mapping theorem.
Nevertheless, one can observe that once we remove the symmetry $\zeta \mapsto -\zeta$, the droplet in the pre-critical case of Theorem~\ref{Thm_insertion at 0} is simply connected.  
This explains the reason why we need the idea of removing symmetry. 
\end{rem}

The rest of this paper is organised as follows. 

\begin{itemize}
    \item In Section~\ref{Section_proof}, we prove Theorems~\ref{Thm_insertion at 0 symmetric} and \ref{Thm_insertion at 0}. 
    In Subsection~\ref{Subsec_post}, we show the post-critical case of Theorem~\ref{Thm_insertion at 0 symmetric} in a more general setup, see Proposition~\ref{Thm_prec}. 
    On the other hand, in Subsection~\ref{Subsec_pre}, we show the pre-critical case of Theorem~\ref{Thm_insertion at 0}. 
    Then by the relation \eqref{droplet relations}, these complete the proof of our main results. 
    \smallskip 
    \item  This article contains two appendices. 
    In Appendix~\ref{Appendix conformal mapping}, we explain the conformal mapping method to derive the ``correct candidate'' of the droplets.
    In Appendix~\ref{Appendix_Hermitian}, we present a way to solve a one-dimensional equilibrium problem in Remark~\ref{Rem_eq msr 1D}, which shares a common feature with the conformal mapping method. 
    These appendices are made only for instructive purposes and the readers who only want the proof of the main theorems may stop at the end of Section~\ref{Section_proof}.
\end{itemize}

\section{Proof of main theorem} \label{Section_proof}

In this section, we show Theorems~\ref{Thm_insertion at 0 symmetric} and \ref{Thm_insertion at 0}.

\subsection{Post critical cases} \label{Subsec_post}

Extending \eqref{Q main}, we consider the potential 
\begin{equation} \label{Q p general}
Q_p(\zeta):=\frac{1}{1-\tau^2}\Big( |\zeta|^2-\tau \re \zeta^2 \Big)-2c\log|\zeta-p|,  \qquad p \in \mathbb{C}. 
\end{equation}
For the case $\tau=0$, the shape of the droplet associated with the potential \eqref{Q p general} was fully characterised in \cite{MR3280250}. 
(In this case, it suffices to consider the case $p \ge 0$ due to the rotational invariance.)
In particular, it was shown that if 
\begin{equation} \label{c p condition tau=0}
	c < \frac{(1-p^2)^2}{4p^2}, \qquad \tau=0, \qquad p \ge 0, 
\end{equation}
the droplet is given by $S=\overline{\mathbb{D}(0,\sqrt{1+c})} \setminus \mathbb{D}(p,\sqrt{c}),$ where $\mathbb{D}(p,R)$ is the disc with centre $p$ and radius $R$, cf. see Remark~\ref{Rem_post cricial tau0} for the other case $c>(1-p^2)^2/(4p^2)$.

To describe the droplets associated with $Q_p$, we denote 
\begin{equation} \label{S1 Qp}
S_1:= \Big\{ (x,y)\in \R^2 : \Big( \frac{x}{1+\tau} \Big)^2+ \Big( \frac{y}{1-\tau} \Big)^2 \le 1+c\Big\}
\end{equation}
and 
\begin{equation} \label{S2 Qp}
S_2:=\Big\{ (x,y)\in \R^2 : (x-\re p)^2+(y-\im p)^2 \le (1-\tau^2)c \Big\}.
\end{equation}
Then we obtain the following. 

\begin{prop} \label{Thm_prec}
Suppose that the parameters $\tau, c \in \R$ and $p \in \mathbb{C}$ are given to satisfy  
\begin{equation} \label{prec_con}
 S_2 \subset S_1, 
\end{equation}
where $S_1$ and $S_2$ are given by \eqref{S1 Qp} and \eqref{S2 Qp}. 
Then the droplet $S \equiv S_{Q_p}$ associated with \eqref{Q p general} is given by
\begin{equation}
S=S_1 \cap (\textup{Int}\,S_2)^c. 
\end{equation} 
\end{prop} 

See Figure~\ref{Fig_droplet gen p} for the shape of the droplets and numerical simulations of Fekete point configurations. 
We remark that with slight modifications, Proposition~\ref{Thm_prec} can be further extended to the case with multiple point charges, i.e. the potential of the form 
\begin{equation} \label{Q multiple pts}
\frac{1}{1-\tau^2}\Big( |\zeta|^2-\tau \re \zeta^2 \Big)-2\sum c_j\log|\zeta-p_j|,  \qquad p_j \in \mathbb{C}, \quad c_j \ge 0.
\end{equation}
(See Remark~\ref{Rem_conformal and higher moments} for a related discussion.) 
Let us also mention that a similar statement for an equilibrium problem on the sphere was shown in \cite{brauchart2018logarithmic}.
For a treatment of a more general case, we refer the reader to \cite{del2019equilibrium,legg2021logarithmic,criado2022vector}. 

\begin{figure}[h!]
	\begin{subfigure}{0.32\textwidth}
		\begin{center}	
			\includegraphics[width=\textwidth]{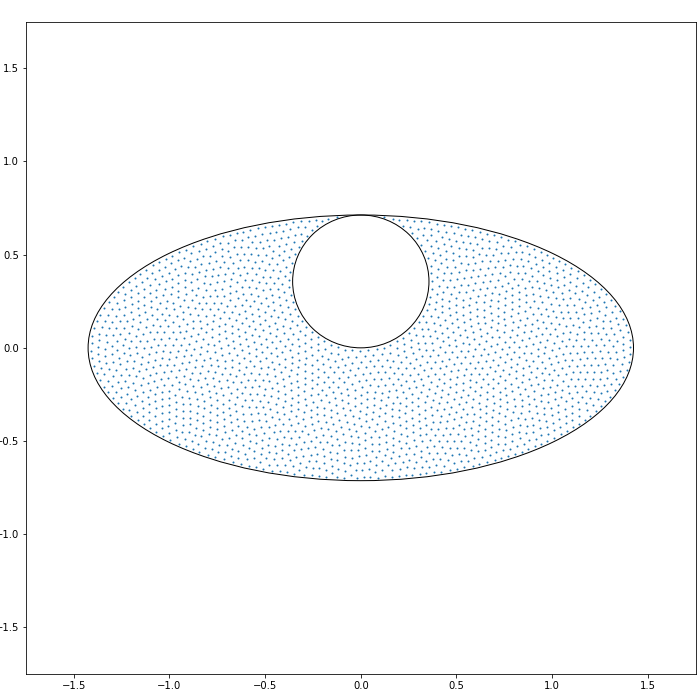}
		\end{center}
		\subcaption{$p=\frac{2}{21} \sqrt{14} \, i$}
	\end{subfigure}	
	\begin{subfigure}{0.32\textwidth}
		\begin{center}	
			\includegraphics[width=\textwidth]{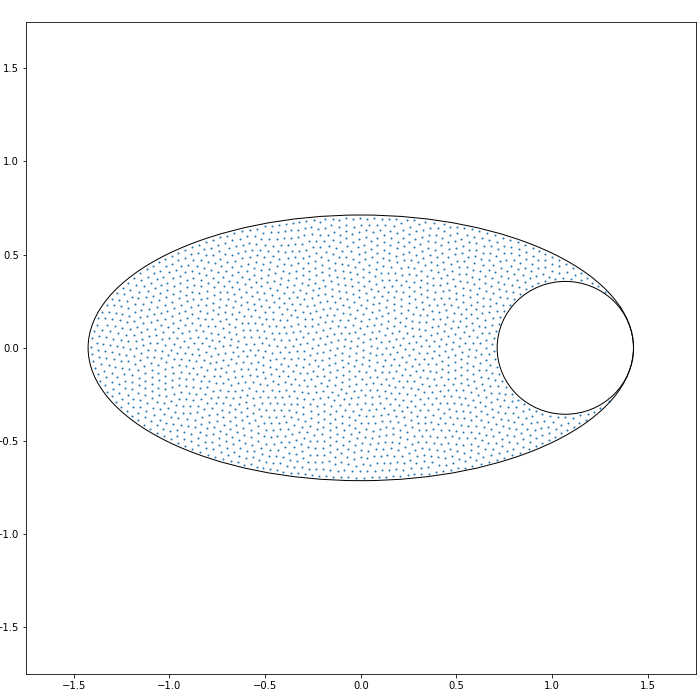}
		\end{center}
		\subcaption{$p=\frac{2}{7}\sqrt{14} $}
	\end{subfigure}	
	\begin{subfigure}{0.32\textwidth}
		\begin{center}	
			\includegraphics[width=\textwidth]{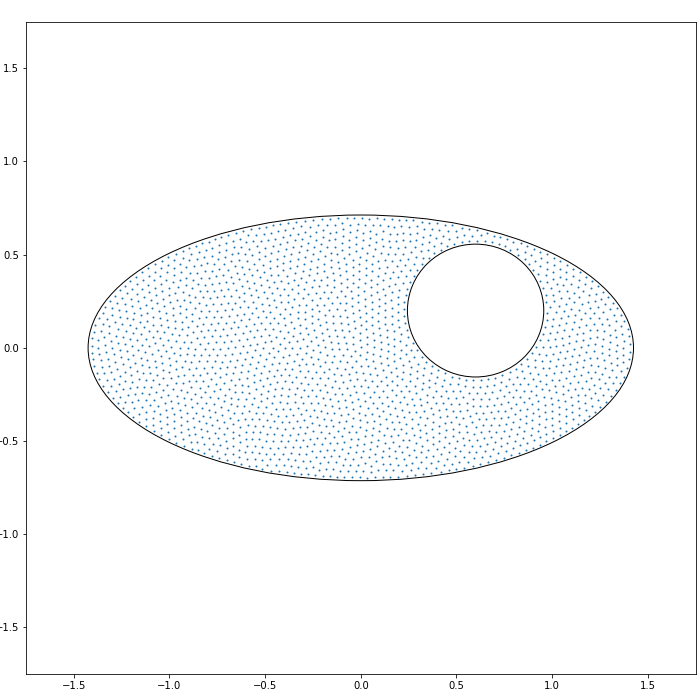}
		\end{center}
		\subcaption{$p=\frac{3}{5}+\frac{1}{5}i$}
	\end{subfigure}	
	\caption{ The droplet $S$ in Proposition~\ref{Thm_prec}, where $\tau=1/3$ and $c=1/7$. Here, a Fekete point configuration with $N=2048$ is also displayed. } \label{Fig_droplet gen p}
\end{figure}

\begin{rem} 
If $p=0$, the condition \eqref{prec_con} corresponds to 
\begin{equation}
\tau < \frac{1}{1+2c}= \tau_c.  
\end{equation}
Therefore Proposition~\ref{Thm_prec} for the special value $p=0$ gives Theorem~\ref{Thm_insertion at 0 symmetric} (i). 
As a consequence, by \eqref{droplet relations}, Theorem~\ref{Thm_insertion at 0} (i) also follows. 
We also mention that if $\tau=0$ and $p>0$, the condition \eqref{prec_con} coincides with \eqref{c p condition tau=0}.
\end{rem}

\begin{rem}[Equilibrium measure in the Hermitian limit] \label{Rem_eq msr 1D}
Before moving on to the planar equilibrium problem for \eqref{Q p general}, we first discuss the one-dimensional problem arising in the Hermitian limit.
For $p \in \R$, the Hermitian limit $\tau \uparrow 1$ of the potential $Q_p$ is given by
\begin{equation} 
\lim_{ \tau \uparrow 1 } Q_p(x+iy)=V_p(x):=\begin{cases}
\displaystyle \frac{x^2}{2}-2c \log|x-p|, &\text{if }y=0,
\smallskip 
\\
+\infty &\text{otherwise}.
\end{cases} 
\end{equation}   
Then one can show that the associated equilibrium measure $\mu_V \equiv \mu_{V_p}$ is given by 
\begin{equation} \label{eq msr 1D}
\frac{d\mu_{V}(x)}{dx}=\frac{\sqrt{ -\prod_{j=1}^{4}(x-\lambda_j) } }{2\pi|x-p|}\cdot  \mathbbm{1}_{ [\lambda_1, \lambda_2 ] \cup [\lambda_3, \lambda_4 ] }(x),
\end{equation}
where 
\begin{align}
& \lambda_1=\frac{ p-2- \sqrt{(p+2)^2+8c}}{2}, \qquad \lambda_2=\frac{ p+2- \sqrt{(p-2)^2+8c}}{2}, \label{lambda 12}
\\
& \lambda_3=\frac{ p-2+ \sqrt{(p+2)^2+8c}}{2}, \qquad \lambda_4=\frac{ p+2+ \sqrt{(p-2)^2+8c}}{2}.  \label{lambda 34}
\end{align}
We remark that when $p=0$, it recovers \eqref{MP square}. See Figure~\ref{Fig_1D eq msrs} for the graphs of the equilibrium measure $\mu_{V_p}$.
The equilibrium measure \eqref{eq msr 1D} follows from the standard method using the Stieltjes transform and the Sokhotski-Plemelj inversion formula. 
For reader's convenience, we provide a proof of \eqref{eq msr 1D} in Appendix~\ref{Appendix_Hermitian}.  

\begin{figure}[h!]
    \begin{subfigure}{0.32\textwidth}
    	\begin{center}	
    		\includegraphics[width=\textwidth]{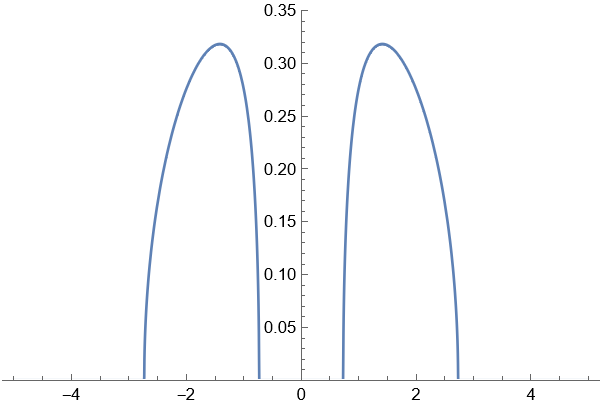}
    	\end{center}
    	\subcaption{$p=0$}
    \end{subfigure}	
    \begin{subfigure}{0.32\textwidth}
    	\begin{center}	
    		\includegraphics[width=\textwidth]{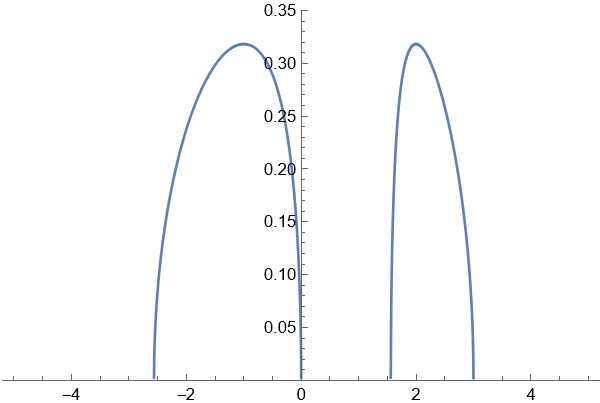}
    	\end{center}
    	\subcaption{$p=1$}
    \end{subfigure}	
    \begin{subfigure}[h]{0.32\textwidth}
    	\begin{center}
    		\includegraphics[width=\textwidth]{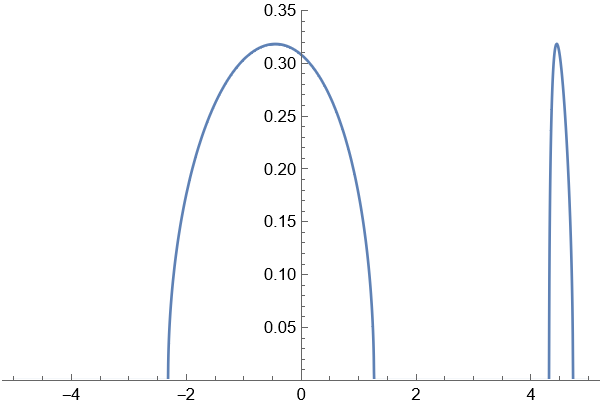}
    	\end{center} \subcaption{$p=4$}
    \end{subfigure}
	\caption{ Graphs of the equilibrium measure $\mu_{V_p}$, where $c=1$. } \label{Fig_1D eq msrs} 
\end{figure} 
\end{rem}

\medskip 

In the rest of this subsection, we prove Proposition~\ref{Thm_prec}. 
First, let us show the following elementary lemmas. 

\begin{lem} \label{Lem_ell var}
For $a, b > 0$, let
$$
K:= \Big\{ (x,y) \in \R^2:  \Big( \frac{x}{a}\Big)^2+\Big( \frac{y}{b} \Big)^2 \le 1 \Big\}.
$$
Then we have 
\begin{equation} \label{Cauchy ell}
\int_K \frac{1}{\zeta-z}\,dA(z)=
\begin{cases}
\displaystyle \bar{\zeta}-\frac{a-b}{a+b}\,\zeta &\text{if }\zeta \in K,
\smallskip
\\
\displaystyle \frac{2ab}{a^2-b^2}\Big(\zeta-\sqrt{\zeta^2-a^2+b^2}\Big) &\text{otherwise.}
\end{cases}
\end{equation}
In particular, for $\zeta \in K$, there exists a constant $c_0 \in \R$ such that 
\begin{equation}
\int_K \log|\zeta-z|^2 \,dA(z)=|\zeta|^2-\frac{a-b}{a+b}\re \zeta^2+c_0. 
\end{equation}
\end{lem}

\begin{rem}
The Cauchy transform in \eqref{Cauchy ell} is useful to explicitly compute the moments of the equilibrium measure. 
Namely, by definition, we have
\begin{equation*}
\int_K \frac{1}{\zeta-z} \, dA(z)= \frac{1}{\zeta} \sum_{k=0}^\infty \frac{1}{\zeta^k} \int_K z^k\,dA(z), \qquad \zeta \to \infty.
\end{equation*}
On the other hand, we have
\begin{align*}
\zeta-\sqrt{\zeta^2-a^2+b^2}= \frac{a^2-b^2}{\zeta} \sum_{k=0}^\infty \binom{1/2}{k+1} \frac{ (b^2-a^2)^{k} }{ \zeta^{2k} }, \qquad \zeta \to \infty.
\end{align*}
Combining the above equations with \eqref{Cauchy ell}, we obtain that for any non-negative integer $k,$ 
\begin{equation}
\frac{1}{ab}\int_K z^{2k} \,dA(z) = 2   \binom{1/2}{k+1} (b^2-a^2)^k.
\end{equation}
\end{rem}

\begin{proof}[Proof of Lemma~\ref{Lem_ell var}]
Recall that $\mathbb{D}$ is the unit disc with centre the origin. 
Then the Joukowsky transform $f: \bar{\mathbb{D}}^c \to K^c$ is given by 
\begin{equation}
f(z)=\frac{a+b}{2}\, z+ \frac{a-b}{2}\, \frac{1}{z}.
\end{equation}	
By applying Green's formula, we have
\begin{equation}
\int_K \frac{1}{\zeta-z}\,dA(z)=\frac{1}{2\pi i} \int_{ \partial K } \frac{\bar{z}}{\zeta-z}\,dz+\bar{\zeta} \cdot  \mathbbm{1}_{ \{\zeta \in K\} }.
\end{equation}
Furthermore, by the change of variable $z=f(w)$, it follows that
\begin{align}
\int_{ \partial K } \frac{\bar{z}}{\zeta-z}\,dz&=\int_{ \partial \mathbb{D} } \frac{ \overline{ f(1/\bar{w})  }   }{\zeta-f(w)} f'(w)\, dw=\int_{ \partial \mathbb{D} } g_\zeta(w)\,dw,
\end{align}
where $g_\zeta$ is the rational function given by 
\begin{equation}
g_\zeta(w):=\frac{1}{\zeta-f(w)}
\Big( \frac{a+b}{2} \frac1w+\frac{a-b}{2} w \Big) \Big( \frac{a+b}{2} - \frac{a-b}{2} \frac{1}{w^2} \Big). 
\end{equation}
Observe that 
$$
\zeta=f(w) \qquad \textrm{if and only if} \qquad w=w_\zeta^\pm:=\frac{\zeta \pm \sqrt{\zeta^2-a^2+b^2} }{a+b},
$$
i.e. 
the points $w_\zeta^\pm$ are solutions to the quadratic equation
$$
(a+b)w^2-2\zeta \, w+(a-b)=0.
$$
Here, the branch of the square root is chosen such that
$$
w_\zeta^- \to 0 \qquad \zeta \to \infty.
$$
By above observation, the function $g_\zeta$ has poles only at 
$$
0, \qquad w_\zeta^+, \qquad w_\zeta^-. 
$$
Moreover note that 
$$
\zeta \in K \qquad \textrm{if and only if} \qquad w_\zeta^\pm \in \mathbb{D}. 
$$
Notice that if $\zeta \in K^c$, then $w_\zeta^- \in \mathbb{D}$ and $w_\zeta^+ \in \mathbb{D}^c$.

Using the residue calculus, we have 
\begin{equation}
\underset{ w=0 }{\textrm{Res}} \,\Big[ g_\zeta(w)\Big]=\frac{a+b}{a-b}\,\zeta.
\end{equation}
On the other hand, we have
\begin{equation}
\underset{ w=w_\zeta^\pm }{\textrm{Res}} \,\Big[ g_\zeta(w) \Big]=- \overline{  f(1/ \bar{w}_\zeta^\pm  ) }=-\frac{a+b}{2} \frac{1}{w_\zeta^\pm }-\frac{a-b}{2} w_\zeta^\pm. 
\end{equation} 
In particular, 
\begin{equation}
\underset{ w=w_\zeta^+ }{\textrm{Res}} \,\Big[ g_\zeta(w) \Big]+\underset{ w=w_\zeta^- }{\textrm{Res}} \,\Big[ g_\zeta(w) \Big] =-2 \frac{a^2+b^2}{a^2-b^2}\,\zeta.  
\end{equation}
Combining all of the above, we obtain the desired identity \eqref{Cauchy ell}.  
The second assertion immediately follows from \eqref{Cauchy ell} and the real-valuedness of $\zeta \mapsto \int_K \log|\zeta-z|^2 \,dA(z)$. 
\end{proof} 

\begin{lem} \label{Lem_disc var}
For $R>0$ and $p \in \mathbb{C}$ we have 
\begin{equation}
\int_{ \mathbb{D}(p,R) } \log|\zeta-z| \,dA(z)=
\begin{cases}
\displaystyle R^2 \log|\zeta-p| &\text{if }\zeta \notin \mathbb{D}(p,R),
\smallskip 
\\
\displaystyle R^2 \log R-\frac{R^2}{2}+\frac{|\zeta-p|^2}{2} &\text{otherwise}.
\end{cases}
\end{equation}
\end{lem}
\begin{proof} 
First, recall the well-known Jensen's formula: for $r>0$, 
\begin{equation} \label{Jensen}
\frac{1}{2\pi} \int_0 ^{2\pi} \log |\zeta-r e^{i\theta}| \,d\theta
=
\begin{cases}
\log r &\text{if }r>|\zeta|,
\smallskip 
\\
\log|\zeta| &\text{otherwise}.
\end{cases}
\end{equation}	
By the change of variables, we have 
\begin{align*}
\int_{ \mathbb{D}(p,R) } \log|\zeta-z| \,dA(z)&= \int_{ \mathbb{D}(0,R) } \log|\zeta-p-z| \,dA(z) =\frac{1}{\pi} \int_0^R r \int_{0}^{2\pi} \log|\zeta-p-r e^{i\theta}|\,d\theta \,dr.  
\end{align*}
Suppose that $\zeta \notin \mathbb{D}(p,R)$. Then by applying \eqref{Jensen}, we have 
\begin{align*}
\frac{1}{\pi} \int_0^R r \int_{0}^{2\pi} \log|\zeta-p-r e^{i\theta}|\,d\theta \,dr&= 2\int_0^R r \, \log|\zeta-p| \,dr=R^2 \log|\zeta-p|.
\end{align*}
On the other hand if $\zeta \in \mathbb{D}(p,R)$, we have 
\begin{align*}
 \frac{1}{\pi} \int_0^R r \int_{0}^{2\pi} \log|\zeta-p-r e^{i\theta}|\,d\theta \,dr
&= 2\int_0^{ |\zeta-p| } r \, \log|\zeta-p| \,dr+2\int_{ |\zeta-p| }^R r \, \log r \,dr
\\
&=R^2 \log R-\frac{R^2}{2}+\frac{|\zeta-p|^2}{2},
\end{align*}
which completes the proof. 
\end{proof} 

We are now ready to complete the proof of Proposition~\ref{Thm_prec}. 

\begin{proof}[Proof of Proposition~\ref{Thm_prec}]
Note that by \eqref{Frostman thm}, the equilibrium measure $\mu$ associated with $Q_p$ is of the form
\begin{equation}
d\mu(z):=\Delta Q_p(z) \cdot  \mathbbm{1}_S(z) \,dA(z)= \frac{1}{1-\tau^2} \cdot  \mathbbm{1}_S(z) \,dA(z).
\end{equation}	
Due to the assumption \eqref{prec_con}, we have
\begin{align*}
\int \log \frac{1}{|\zeta-z|^2}\,d\mu(z) =\frac{1}{1-\tau^2} \Big( \int_{S_1} \log \frac{1}{|\zeta-z|^2}\,dA(z)- \int_{S_2} \log \frac{1}{|\zeta-z|^2}\,dA(z) \Big).
\end{align*}
Note that by Lemma~\ref{Lem_ell var}, there exists a constant $c_0$ such that 
\begin{equation} \label{S1 int Qp}
\int_{S_1} \log \frac{1}{|\zeta-z|^2}\,dA(z)=-|\zeta|^2+\tau \re \zeta^2-c_0. 
\end{equation}
On the other hand, by Lemma~\ref{Lem_disc var}, we have 
\begin{equation}  \label{S2 int Qp}
\int_{S_2} \log \frac{1}{|\zeta-z|^2}\,dA(z)=-2(1-\tau^2)c \, \log|\zeta-p|.  
\end{equation}
Combining \eqref{S1 int Qp}, \eqref{S2 int Qp} and \eqref{Q p general}, we obtain
\begin{equation}
\int \log \frac{1}{|\zeta-z|^2}\,d\mu(z) =-Q_p(\zeta)-\frac{c_0}{1-\tau^2},
\end{equation}
which leads to \eqref{vari eq_prec}. 

Next, we show the variational inequality \eqref{vari ineq_prec}. Note that if $\zeta \in S_2$, it immediately follows from Lemma~\ref{Lem_disc var}. Thus it is enough to verify \eqref{vari ineq_prec} for the case $\zeta \in S_1^c$. 
Let 
\begin{equation}
H_p(\zeta):= \int \log \frac{1}{|\zeta-z|^2} \,d\mu(z)+Q_p(\zeta). 
\end{equation}
Suppose that the variational inequality \eqref{vari ineq_prec} does not hold. 
Then since $H_p(\zeta) \to \infty$ as $\zeta \to \infty$, there exists $\zeta_* \in S_1^c$ such that 
\begin{equation}
 \partial_\zeta H_p(\zeta)|_{\zeta=\zeta_*} = 0. 
\end{equation}
On the other hand, by Lemmas~\ref{Lem_ell var} and \ref{Lem_disc var}, if $\zeta \in S_1^c$, the Cauchy transform of the measure $\mu$ is computed as 
\begin{equation} \label{Cauchy transform prec}
\int \frac{d\mu(z)}{\zeta-z} =\frac{1}{2\tau}\Big(\zeta-\sqrt{\zeta^2-4\tau(1+c)} \Big)-\frac{c}{\zeta-p}.
\end{equation}
Together with \eqref{Q p general}, this gives rise to  
\begin{align}
\partial_\zeta Q_p(\zeta)-\int \frac{d\mu(z)}{\zeta-z}=\frac{1}{1-\tau^2}\Big(\bar{\zeta}-\tau \zeta\Big)-\frac{1}{2\tau}\Big(\zeta-\sqrt{\zeta^2-4\tau(1+c)}\Big).
\end{align}
Then it follows that the condition $\partial_\zeta H_p(\zeta)=0$ is equivalent to 
\begin{equation}
(1+\tau^2)|\zeta|^2-\tau(\zeta^2+\bar{\zeta}^2)=(1-\tau^2)^2(1+c).
\end{equation}
Therefore, by \eqref{S1 Qp}, one can notice that $\partial_\zeta H_p(\zeta)=0$ if and only if $\zeta \in \partial S_1$.
This yields a contradiction with the assumption $\zeta_* \in S_1^c.$
Therefore we conclude that the variational inequality \eqref{vari ineq_prec} holds for $\zeta \in S^c$, which completes the proof. 
\end{proof} 

\begin{rem}
Let us denote by 
\begin{equation}
m_k:= \int z^k \,d\mu(z)
\end{equation}
the $k$-th moment of the equilibrium measure. 
Notice that the Cauchy transform of $\mu$ satisfies the asymptotic expansion 
\begin{equation} \label{Cauchy transform and moments}
\int \frac{d\mu(z)}{\zeta-z} =  \frac{1}{\zeta} \sum_{k=0}^\infty  \frac{m_k}{\zeta^k} , \qquad \zeta \to \infty. 
\end{equation}
Using this property and \eqref{Cauchy transform prec}, after straightforward computations, one can verify that the equilibrium measure $\mu$ in Proposition~\ref{Thm_prec} has the moments
\begin{equation}
m_{2k}= 2 \frac{ (2k-1)! }{  (k-1)! (k+1)! } \tau^{k}   (1+c)^{k+1} - c\, p^{2k}, \qquad m_{2k+1} = -c\, p^{2k+1}. 
\end{equation}
Notice in particular that if $p=0$, all odd moments vanish. 
\end{rem}

\subsection{Pre-critical case}    \label{Subsec_pre}

In this subsection, we show Theorem~\ref{Thm_insertion at 0} (ii). 
Then by \eqref{droplet relations}, Theorem~\ref{Thm_insertion at 0 symmetric} (ii) follows.

\begin{proof}[Proof of Theorem~\ref{Thm_insertion at 0} (ii)]
Recall that $\widehat{Q}$ is given by \eqref{Q hat} and that all we need to show is the variational principles \eqref{vari eq_prec} and \eqref{vari ineq_prec} for $W=\widehat{Q}.$ 
For this, similar to above, let 
\begin{equation}
H(\zeta):= \int \log \frac{1}{|\zeta-z|^2} \,d\widehat{\mu}(z)+ \widehat{Q}(\zeta),
\end{equation}
where $\widehat{\mu}$ is the equilibrium measure associated with $\widehat{Q}$. 
Then 
\begin{equation}
\partial_\zeta H(\zeta)= \partial_\zeta \widehat{Q}(\zeta)-C(\zeta)= \frac{1}{1-\tau^2} \Big( \sqrt{ \frac{ \bar{\zeta} }{ \zeta } } - \tau \Big) -\frac{c}{\zeta}-C(\zeta),  
\end{equation}
where $C(\zeta)$ is the Cauchy transform of $\widehat{\mu}$ given by 
\begin{equation}
C(\zeta) =\frac{1}{2(1-\tau)^2} \int_{ \widehat{S} } \frac{1}{\zeta-z} \frac{1}{|z|}\,dA(z).  
\end{equation}
Here, we have used \eqref{Frostman thm}. 

Applying Green's formula, we have 
\begin{align} 
\begin{split} \label{Cauchy Green}
(1-\tau^2)C(\zeta) &= \frac{1}{2\pi i}   \int_{\partial \widehat{S} } \frac{1}{\zeta-z} \sqrt{ \frac{\bar{z}}{z} } \, dz + \sqrt{  \frac{\bar{\zeta}}{\zeta}  } \cdot  \mathbbm{1}_{ \{\zeta \in \Int( \widehat{S} ) \} } . 
\end{split}
\end{align}
Recall that $f$ is given by \eqref{f rational main}. 
Let
\begin{align}
\begin{split} \label{def of g}
g(w)&:=   \sqrt{ \frac{\overline{f(1/\bar{w})} }{f(w)} } \,f'(w). 
\end{split}
\end{align}
Since $f'(a\tau)=0,$ the function $g(w)$ has poles only at $0, 1/a, a$.
We also write 
\begin{equation}
h_\zeta(w):= \frac{g(w)}{\zeta-f(w)} .
\end{equation}
Using the change of variable $z=f(w)$, 
\begin{align}
\begin{split}
 \frac{1}{2\pi i}   \int_{\partial \widehat{S} } \frac{1}{\zeta-z} \sqrt{ \frac{\bar{z}}{z} } \, dz &= \frac{1}{2\pi i} \int_{\partial \mathbb{D}} \frac{1}{\zeta-f(w)} \sqrt{ \frac{\overline{f(1/\bar{w})} }{f(w)} } \,f'(w) \, dw 
 = \frac{1}{2\pi i} \int_{\partial \mathbb{D}} h_\zeta(w) \, dw .
\end{split}
\end{align}
By the residue calculus, we have 
\begin{equation}
\underset{ w=0 }{\textrm{Res}} \,\Big[ h_\zeta(w) \Big] = \frac{1}{\tau}, \qquad \underset{ w=a }{\textrm{Res}} \,\Big[ h_\zeta(w) \Big] = 0. 
\end{equation}

Note that $\zeta=f(w)$ is equivalent to 
\begin{equation} \label{cubic equation v0}
d(1-aw)(w-a\tau)^2= w(w-a)\zeta, \qquad d=\frac{(1+\tau)(1+2c)}{2},
\end{equation}
which can be rewritten as a cubic equation 
\begin{equation} \label{cubic equation w zeta}
ad w^3-( d+2a^2d\tau-\zeta  )w^2+a(2d\tau+a^2\tau^2 d-\zeta)w-a^2\tau^2 d= 0. 
\end{equation}

For given $\zeta \in \mathbb{C}$, there exist $w_\zeta^{(j)}$ $(j=1,2,3)$ such that $f(w_\zeta^{(j)})=\zeta.$
Note that by \eqref{cubic equation w zeta}, we have 
\begin{equation}
w_\zeta^{(1)} w_\zeta^{(2)} w_\zeta^{(3)}= a \tau^2 \in (-1,0). 
\end{equation}
Furthermore, since $f$ is a conformal map from $\mathbb{D}^c$ onto $\widehat{S}^c$, we have the following:
\begin{enumerate}
    \item If $\zeta \in \Int( \widehat{S} )$, then all $w_\zeta^{(j)}$'s are in $\mathbb{D}$; 
    \smallskip 
    \item If $\zeta \in \widehat{S}^c$, then two of $w_\zeta^{(j)}$'s are in $\mathbb{D}$. 
\end{enumerate}
By the residue calculus using \eqref{f rational main} and \eqref{def of g}, for each $j,$ 
\begin{align}
\begin{split}
\underset{ w=w_\zeta^{(j)} }{\textrm{Res}} \,\Big[ h_\zeta(w) \Big]&= -\frac{ g(w_\zeta^{(j)}  ) }{ f'( w_\zeta^{(j)} )  }  = - \frac{(w^{(j)}_\zeta-a)(1-a\tau w^{(j)}_\zeta)}{ (w^{(j)}_\zeta-a\tau)(1-aw^{(j)}_\zeta) }
\\
&= \frac{d}{\zeta} \frac{ (a\tau w_\zeta^{(j)}-1) ( w_\zeta^{(j)}-a\tau  )   }{ w_\zeta^{(j)} } 
= \frac{d}{\zeta} \Big(  a \tau  w_\zeta^{(j)}-(a^2\tau^2+1)  + \frac{ a\tau   }{ w_\zeta^{(j)} } \Big), 
\end{split}
\end{align}
where we have used \eqref{cubic equation v0}. 
On the other hand, it follows from \eqref{cubic equation w zeta} that  
\begin{equation}
\sum_{j=1}^3 w_\zeta^{(j)}= \frac{ d+2a^2 d\tau-\zeta }{ ad }, \qquad  
\sum_{j=1}^3 \frac{1}{w_\zeta^{(j)}}= \frac{-\zeta+2d\tau+a^2\tau^2d}{a\tau^2 d}. 
\end{equation}
These relations give rise to 
\begin{align}
\begin{split}
d\sum_{j=1}^3 \Big(  a \tau  w_\zeta^{(j)}-(a^2\tau^2+1)  + \frac{ a\tau   }{ w_\zeta^{(j)} } \Big)&= d\tau+2a^2d\tau^2-\tau \zeta-\frac{\zeta}{\tau} +2d+a^2\tau d-3d(a^2\tau^2+1) 
\\
&= -\Big( \tau+\frac{1}{\tau} \Big) \zeta -c(1-\tau^2).
\end{split}
\end{align}
Combining all of the above, we have shown that if $\zeta \in \Int (\widehat{S})$, 
\begin{equation}
\sum_{j=1}^3\underset{ w=w_\zeta^{(j)} }{\textrm{Res}} \,\Big[ h_\zeta(w) \Big] = -\Big( \tau+\frac{1}{\tau} \Big)- \frac{c(1-\tau^2)}{\zeta}. 
\end{equation}
Therefore if $\zeta \in \Int (\widehat{S})$, we obtain 
\begin{align}
\begin{split}
(1-\tau^2)C(\zeta) &=  \sqrt{  \frac{\bar{\zeta}}{\zeta}  } -\frac{1}{\tau} -c(1-\tau^2) = (1-\tau^2) \partial_\zeta\widehat{Q}(\zeta).
\end{split}
\end{align}
Then by \eqref{Cauchy Green}, the variational equality \eqref{vari eq_prec} follows. 

Now it remains to show the variational inequality \eqref{vari ineq_prec}. 
Note that by definition, $H(\zeta) \to \infty$ as $\zeta \to \infty$. 
Suppose that the variational inequality \eqref{vari ineq_prec} does not hold. 
Then there exists $\zeta_* \in \widehat{S}^c$ such that 
\begin{equation} \label{contra eq}
\partial_\zeta H(\zeta)|_{\zeta=\zeta_* }= \partial \widehat{Q}(\zeta_*)-C(\zeta_*) =0.  
\end{equation}
Recall that if $\zeta \in \widehat{S}^c$, then only one of $w_\zeta^{(j)}$'s, say $w_\zeta$, is in $\mathbb{D}^c$. 
By combining the above computations, we have that for $\zeta \in \widehat{S}^c$,  
\begin{align}
\begin{split}
(1-\tau^2) \Big( \partial_\zeta \widehat{Q}(\zeta)-C(\zeta) \Big) = \sqrt{ \frac{\bar{\zeta}}{\zeta} } -  \underset{ w=w_\zeta }{\textrm{Res}} \,\Big[ h_\zeta(w) \Big] =  \sqrt{ \frac{\bar{\zeta}}{\zeta} }-  \frac{d}{\zeta} \Big(  a \tau  w_\zeta-(a^2\tau^2+1)  + \frac{ a\tau   }{ w_\zeta } \Big).
\end{split}
\end{align}
Therefore the identity \eqref{contra eq} holds if and only if 
\begin{equation} \label{zeta * eq}
|\zeta_*|  = d  \Big(  a \tau  w_{\zeta_*}-(a^2\tau^2+1)  + \frac{ a\tau   }{ w_{\zeta_*} } \Big).
\end{equation} 
Note that by \eqref{f rational main},
\begin{align*}
-\frac{d}{f(x)}  \Big(  a \tau  x-(a^2\tau^2+1)  + \frac{ a\tau   }{ x } \Big) &= -\frac{1}{f(x)}\frac{(1+\tau)(1+2c)}{2} \frac{ (a\tau x-1)(x-a\tau) }{x} = \frac{ (a\tau x-1)(x-a) }{ (a x -1)(x-a \tau) }. 
\end{align*}
Therefore if $x<1/(a\tau)$, 
\begin{align*}
d \Big(  a \tau  x-(a^2\tau^2+1)  + \frac{ a\tau   }{ x } \Big) < \tau |f(x)| < |f(x)|. 
\end{align*}
From this, we notice that \eqref{zeta * eq} does not hold for $w_{ \zeta_* } \in \R$. 
Furthermore, this implies that the right-hand side of \eqref{zeta * eq} is real-valued if and only if $w_{\zeta_*} \in \partial \mathbb{D}$, equivalently, $\zeta_* \in \partial \widehat{S}.$ 
This contradicts with the assumption that $\zeta_* \in \widehat{S}^c$. 
Now the proof is complete.   
\end{proof}

\appendix

\section{Conformal mapping method: the pre-critical case} \label{Appendix conformal mapping}

In this appendix, we present the conformal mapping method, which is helpful to derive the candidate of the droplet given in terms of the rational function \eqref{f rational main}.

\begin{prop} \label{Prop_conformal map}
Let $\tau \in (\tau_c,1)$.  
Suppose that $\widehat{S}$ in \eqref{eq msr Q hat} is simply connected. 
Let $f$ be a unique conformal map $(\bar{\mathbb{D}}^c,\infty) \to (\widehat{S}^c,\infty)$, which satisfies
\begin{equation}
f(z)=r_1\, z +r_2+O\Big( \frac{1}{z} \Big), \qquad z \to \infty. 
\end{equation}
Then the following holds.
\begin{enumerate}
    \item[(i)] The conformal map $f$ is a rational function of the form 
\begin{equation} \label{f rational form}
f(z)=r_1z+r_2+\frac{r_3}{z}+\frac{r_4}{z-a}, \qquad a\in(-1,0), 
\end{equation}
which satisfies
\begin{equation} \label{f(1/a)=0}
f(1/a)=\frac{r_1}{a}+r_2+r_3 a + \frac{ar_4}{1-a^2}=0. 
\end{equation}
\item[(ii)] The parameters $r_j$ $(j=1,\dots,4)$ are given by
\begin{align} 
\begin{split} \label{r1 r2 r3 r4}
&r_1= \frac{1+\tau}{2} \sqrt{ \frac{1+2c}{\tau} }, \qquad \quad r_2=\frac{1+\tau}{2\tau}( \tau(1+2c)+2\tau-1 ), 
\\
&r_3= \frac{1+\tau}{2}\tau \sqrt{ \tau(1+2c) }, \qquad r_4= \frac{ (1-\tau)^2(1+\tau)(1-(1+2c)\tau) }{ 2\tau\sqrt{\tau(1+2c)} }
\end{split}
\end{align}
and 
\begin{equation}
a= - \frac{1}{ \sqrt{ \tau(1+2c) } }.
\end{equation}
\end{enumerate}
\end{prop}

Note that the rational function $f$ with the choice of parameters \eqref{r1 r2 r3 r4} corresponds to \eqref{f rational main}. 
Therefore Proposition~\ref{Prop_conformal map} gives rise to Theorem~\ref{Thm_insertion at 0} (ii) under the assumption that $\widehat{S}$ is simply connected. 
However, there is no general theory characterising the connectivity of the droplet. 
(Nevertheless, we refer the reader to \cite{MR3454377,lee2014sharpness} for sharp connectivity bounds of the droplets associated with a class of potentials.)
Thus we need to directly verify the variational principles as in Subsection~\ref{Subsec_pre}. 

\begin{proof}[Proof of Proposition~\ref{Prop_conformal map} (i)]
By differentiating the variational equality \eqref{vari eq_prec}, we have
\begin{equation}
\partial_\zeta \widehat{Q}(\zeta)=C(\zeta):=\int \frac{d\widehat{\mu}(z)}{\zeta-z}  , \qquad \zeta \in \widehat{S}. 
\end{equation}	
Using \eqref{Q hat}, this can be rewritten as
\begin{equation}
\bar{\zeta}=\zeta \Big[  (1-\tau^2)  \Big( C(\zeta)+\frac{c}{\zeta} \Big)+\tau   \Big]^2.
\end{equation}
Therefore the Schwarz function $F$ associated with the droplet $\widehat{S}$ exists. 
Furthermore, it is expressed in terms of $C$ as 
\begin{equation} 
F(\zeta)=\zeta \Big[  (1-\tau^2)  \Big( C(\zeta)+\frac{c}{\zeta} \Big)+\tau   \Big]^2.
\end{equation}

Note that for $z \in \partial \mathbb{D},$
\begin{equation}
\overline{f( 1/\bar{z} )}=\overline{f(z)}
= f(z) \Big[  (1-\tau^2)  \Big( C(f(z))+\frac{c}{f(z)} \Big)+\tau  \Big]^2.
\end{equation}
Using this, we define $f: \bar{\mathbb{D}}\backslash \{0\} \to \mathbb{C}$ by analytic continuation as 
\begin{equation} \label{f analytic conti}
f(z):=\overline{  f(1/\bar{z}) \Big[  (1-\tau^2)  \Big( C(f(1/\bar{z}))+\frac{c}{f(1/\bar{z})} \Big)+\tau \Big]^2  }. 
\end{equation}
Therefore $f$ has simple poles only at $0, \infty$ and the point $a \in \R$ such that $f(1/a)=0$, which leads to \eqref{f rational form}. 
\end{proof}

Next, we need to specify the constants $r_j$ and $a.$
For this, we shall find interrelations among the parameters. 

\begin{lem} \label{Lem_r1 r2 r3 r4 relations}
We have 
\begin{equation} \label{r3 r1}
r_3=r_1 \tau^2
\end{equation}
and 
\begin{equation}  \label{r4 r2}
r_4=a(1-\tau^2) \Big( r_2-2\tau(1+c) \Big). 
\end{equation}
Furthermore, we have
\begin{equation}  \label{r2 r1}
r_2= r_1 \frac{1+a^2\tau^2}{1-a^2\tau^2} \frac{a^2-1}{a} +\frac{ 2a^2(1-\tau^2)\tau (1+c) }{1-a^2\tau^2}. 
\end{equation}
\end{lem}

\begin{proof}

Note that 
\begin{equation}
\overline{f(1/\bar{z})}=\frac{r_1}{z}+r_2+r_3 z+\frac{r_4z}{1-az}.
\end{equation}
Therefore, we have
\begin{equation}
\frac{1}{ \overline{f(1/\bar{z})} }= \frac{1}{r_1}\,z-\frac{r_2}{r_1^2} \, z^2+\frac{ r_2^2-r_1r_3-r_1r_4 }{ r_1^3 } \, z^3  +O(z^4), \qquad z \to 0. 
\end{equation}

Since the Cauchy transform $C$ satisfies the asymptotic behaviour
\begin{equation}
C(\zeta)=\frac{1}{\zeta}+O(\frac{1}{\zeta^2}), \qquad \zeta \to \infty,
\end{equation}
we have
\begin{equation}
\overline{ C( f(1/\bar{z}) ) }=\frac{1}{r_1}\,z+O(z^2), \qquad z \to 0. 
\end{equation}
Combining these equations with \eqref{f analytic conti}, we obtain
\begin{align} \label{f asymp 0 v1}
f(z) = \frac{r_1\tau^2}{z}+ \Big( r_2\tau^2+2\tau(1-\tau^2)(1+c) \Big)+O(z), \qquad z \to 0. 
\end{align}
On the other hand, by using \eqref{f rational form}, we have
\begin{equation} \label{f asymp 0 v2}
f(z)=\frac{r_3}{z}+\Big( r_2-\frac{r_4}{a} \Big)+O(z), \qquad z \to 0.  
\end{equation}
Then by comparing the coefficients in \eqref{f asymp 0 v1} and \eqref{f asymp 0 v2}, we obtain \eqref{r3 r1} and \eqref{r4 r2}.

Note that by \eqref{f(1/a)=0}, we have 
\begin{equation}
r_4=\frac{a^2-1}{a}\Big( \frac{r_1}{a}+r_2+r_3 a \Big).
\end{equation}
Then by \eqref{r3 r1}, we have
\begin{equation}
r_4=\frac{a^2-1}{a}\Big( \frac{r_1}{a}+r_2+r_1 a \tau^2\Big)= r_1 \frac{(1+a^2\tau^2)(a^2-1)}{a^2}+r_2 \frac{a^2-1}{a}.
\end{equation}
Combining this identity with \eqref{r4 r2}, we obtain 
\begin{equation}
a(1-\tau^2) r_2-2a(1-\tau^2)\tau(1+c)= r_1 \frac{(1+a^2\tau^2)(a^2-1)}{a^2}+r_2 \frac{a^2-1}{a},
\end{equation}
which leads to \eqref{r2 r1}. 
\end{proof}

\begin{lem} \label{Lem_r1 r2 eq}
We have 
\begin{equation}
\Big( (2-a^2+a^4\tau^2)r_1+ar_2 \Big) \Big(  r_2-2\tau(1+c) \Big) = (1-\tau^2) c^2 a(a^2-1) .
\end{equation}
\end{lem}
\begin{proof}
Using \eqref{f(1/a)=0}, we have
\begin{equation}
\frac{1}{ \overline{f(1/\bar{z})} }= \frac{a^2(a^2-1) }{ (2-a^2)r_1+ar_2+a^4 r_3 }\,\frac{1}{z-a} +O(1), \qquad z \to a.
\end{equation}
Then by \eqref{f analytic conti} and \eqref{r3 r1}, we obtain 
\begin{equation}
r_4=  \frac{(1-\tau^2)^2 c^2 \, a^2(a^2-1) }{ (2-a^2)r_1+ar_2+a^4 r_3 }= \frac{(1-\tau^2)^2 c^2 \, a^2(1-a^2)^2}{ r_1(a^2\tau^2-1)(1-a^2)^2 +r_4 a^2  }.
\end{equation}
Now lemma follows from \eqref{r4 r2}. 
\end{proof}

\begin{proof}[Proof of Proposition~\ref{Prop_conformal map} (ii)]
Since $\widehat{\mu}$ is a probability measure, we have 
\begin{equation} \label{mass 1 condition}
1=\int_{\widehat{S}} \frac{1}{2(1-\tau^2)} \frac{1}{|z|} \, dA(z)=\frac{1}{2\pi i} \int_{\partial \widehat{S}}  \frac{1}{1-\tau^2} \sqrt{ \frac{\bar{z}}{z} } \, dz,
\end{equation}
where we have used Green's formula for the second identity. 
Using the change of variable $z=f(w)$, where $f$ is of the form \eqref{f rational form}, this can be rewritten as 
\begin{equation} \label{mass 1 var}
\frac{1}{2\pi i} \int_{\partial \mathbb{D}}  \sqrt{ \overline{f(1/\bar{w})} f(w) } \,\frac{f'(w)}{f(w)} \, dw =1-\tau^2.  
\end{equation}

By Lemma~\ref{Lem_r1 r2 r3 r4 relations} and \eqref{f rational form}, we have 
\begin{align}
f(z)&= \frac{1-az}{z(z-a)}\Big( -\frac{r_1}{a} z^2+ \Big( \frac{a^2-1}{a^2}r_1-\frac{r_2}{a} \Big) z -a \tau^2  r_1 \Big),
\\
\overline{f(1/\bar{z})}&= \frac{ z-a }{z(1-az)} \Big( -a \tau^2 r_1  z^2+ \Big( \frac{a^2-1}{a^2}r_1-\frac{r_2}{a} \Big) z-\frac{r_1}{a} \Big). 
\end{align}
Note here that by construction, two zeros of $f$ other than $1/a$ are contained in the unit disc.  
Using these together with straightforward residue calculus, we obtain  
\begin{equation} \label{Residue at 0}
\textup{Res}_{w=0} \Big[  \sqrt{ \overline{f(1/\bar{w})} f(w) } \,\frac{f'(w)}{f(w)} \Big] =  (1+c)(1-\tau^2)
\end{equation}
and 
\begin{equation} 
\textup{Res}_{w=a} \Big[  \sqrt{ \overline{f(1/\bar{w})} f(w) } \,\frac{f'(w)}{f(w)} \Big] = -\frac{1}{a} \Big[ \Big(   \frac{1+a^2\tau^2}{a}r_1+r_2  \Big)  \Big(  \frac{a^4\tau^2-a^2+2}{a}r_1+r_2 \Big) \Big]^{1/2}.
\end{equation}
Furthermore, it follows from Lemma~\ref{Lem_r1 r2 eq} that 
\begin{equation} \label{Residue at a}
\textup{Res}_{w=a} \Big[  \sqrt{ \overline{f(1/\bar{w})} f(w) } \,\frac{f'(w)}{f(w)} \Big] = - c(1-\tau^2).
\end{equation}
Combining \eqref{mass 1 var}, \eqref{Residue at 0} and \eqref{Residue at a}, one can notice that the function $f$ has a double zero, which implies that 
\begin{equation} \label{r2 r1 v2}
\frac{a^2-1}{a^2}r_1-\frac{r_2}{a}  =2 r_1 \tau. 
\end{equation}
By solving the system of equations given in Lemmas~\ref{Lem_r1 r2 r3 r4 relations}, ~\ref{Lem_r1 r2 eq} and \eqref{r2 r1 v2}, the desired result follows.  
\end{proof}

\begin{rem}[The use of higher moments of the equilibrium measure] \label{Rem_conformal and higher moments}
In a more complicated case, for instance for the case with multiple point charges such as \eqref{Q multiple pts}, the mass-one condition \eqref{mass 1 condition} may not be enough to characterise the parameters.  
In this case, one can further use the higher order asymptotic expansions appearing in the above lemmas, which involve the $k$-th moments of the equilibrium measure; cf. \eqref{Cauchy transform and moments}. 
Thus in principle, one can always find enough (algebraic) interrelations to characterise the parameters appearing in the conformal map. 
\end{rem}

\begin{rem} \label{Rem_post cricial tau0}
For the case $\tau=0$ and $p >0$, it was shown in \cite{MR3280250} that if
$$
c> \frac{(1-p^2)^2}{4p^2},
$$
the droplet associated with \eqref{Q p general} is a simply connected domain whose boundary is given by the image of the conformal map 
	$$
	f(z)=R\,z-\frac{\kappa}{z-q}-\frac{\kappa}{q}, \qquad R=\frac{1+p^2q^2}{2p q}, \qquad \kappa=\frac{(1-q^2)(1-p^2q^2)}{2pq}.
	$$
	Here, $q$ is given by the unique solution of $P(q^2)=0$, where 
	$$
	P(x):=x^3-\Big( \frac{p^2+4c+2}{2p^2} \Big) x^2+\frac{1}{2p^4}
	$$
	such that $0<q<1$ and $\kappa>0.$
 
Beyond the case $\tau=0$, the conformal mapping method described above also works for the potential \eqref{Q p general} with general $\tau \in [0,1), c \in \R$ and $p \in \mathbb{C}$ under the assumption that the associated droplet is simply connected. 
Under this assumption, one can show that the boundary of the droplet is given by the image of the rational conformal map $f$ of the form
\begin{equation} 
f(z)=R_1 \, z+R_2+\frac{R_3}{z}+\frac{R_4}{z-q}, \qquad q \in \mathbb{D},
\end{equation}
which satisfies $f(1/q)=0$. 
Furthermore, following the strategy above, one can characterise the coefficients $R_j$ ($j=1,\dots, 4$) of this rational map as well as the position of the pole $q \in \mathbb{D}$. 

However, as previously mentioned, it is far from being obvious to characterise a condition for which the droplet is simply connected. 
Nevertheless, since the radius of curvature of the ellipse \eqref{S1 Qp} at the point $(1+\tau)\sqrt{1+c}$ is given by 
$$
\frac{(1-\tau)^2}{1+\tau}\,\sqrt{1+c},
$$
one can expect that if
\begin{equation}
p> \max\Big\{ \frac{4\tau}{1+\tau}\sqrt{1+c} \, , \, (1+\tau) \sqrt{1+c}-\sqrt{1-\tau^2}\sqrt{c} \Big\}
\end{equation}
then the droplet is a simply connected domain.
\end{rem}

\section{One-dimensional equilibrium measure problem in the Hermitian limit} \label{Appendix_Hermitian}

In this appendix, we present a proof of \eqref{eq msr 1D}.
Let us write 
\begin{equation}
V(z) \equiv V_p(z)= \frac{z^2}{2}-2c \log|z-p|. 
\end{equation} 
Recall that $\mu_{V} \equiv \mu_{V_p}$ is the equilibrium measure associated with $V_p(x)$ ($x \in \R$).

We define 
\begin{equation}
R(z):=\Big(\frac{V'(z)}{2}\Big)^2- \int_{\R} \frac{V'(z)-V'(s)}{z-s}\,d\mu_V(s).  
\end{equation}
By applying Schiffer variations (see e.g. \cite[Section 3]{del2019equilibrium}), we have
\begin{equation} \label{R var}
R(z)=\Big( \int \frac{d\mu_V(s)}{z-s}-\frac{V'(z)}{2} \Big)^2, \qquad z \in \mathbb{C}\setminus \supp (\mu_V).
\end{equation}
Combining the asymptotic behaviour
$$
\int \frac{d\mu_V(s)}{z-s} \sim \frac{1}{z}, \qquad z \to \infty,
$$
with \eqref{R var}, we obtain 
\begin{equation} \label{R(z) inf}
R(z)=\frac14 z^2-(c+1)-\frac{cp}{z}+O\Big(\frac{1}{z^2}\Big) , \qquad z \to \infty.
\end{equation}

On the other hand, since
\begin{align*}
V'(z)=z-\frac{2c}{z-p}, \qquad \frac{V'(z)-V'(s)}{z-s}=1+\frac{2c}{z-p} \frac{1}{s-p},
\end{align*}
we have 
\begin{align} \label{R eq0}
R(z)=\frac14 \Big( z-\frac{2c}{z-p} \Big)^2-1-\frac{2c}{z-p} \int_{\R} \frac{d\mu_V(s)}{s-p}.
\end{align}
Thus we obtain 
\begin{equation}\label{R(z) p}
R(z)=\frac{c^2}{(z-p)^2}+O\Big(\frac{1}{z-p} \Big),  \qquad z \to p.
\end{equation}

In the expression \eqref{R eq0}, one can observe that $R$ is a rational function with a double pole at $z=p$. Therefore it is of the form
\begin{equation} \label{R structure}
R(z)=\frac14 z^2+\frac{Az^2+Bz+C}{(z-p)^2}
\end{equation}
for some constants $A,B$ and $C$.
As in the previous subsection, we need to specify these parameters. 

By direct computations, we have 
\begin{equation} \label{R(z) inf1}
R(z)=\frac14 z^2+A+\frac{2Ap+B}{z}+O\Big(\frac{1}{z^2}\Big), 
\qquad z \to \infty,
\end{equation}
and
\begin{equation}\label{R(z) p1}
R(z)=\frac{Ap^2+Bp+C}{(z-p)^2}+O\Big( \frac{1}{z-p}\Big), \qquad z \to p.
\end{equation}
By comparing coefficients in \eqref{R(z) inf} and \eqref{R(z) inf1}, we have 
\begin{equation}
A=-c-1, \qquad -cp=2Ap+B.
\end{equation}
Similarly, by \eqref{R(z) p} and \eqref{R(z) p1}, 
\begin{equation}
Ap^2+Bp+C=c^2.
\end{equation}
By solving these algebraic equations, we obtain 
\begin{equation}
B=p(c+2), \qquad C=c^2-p^2.
\end{equation}

Combining all of the above with \eqref{R structure}, we have shown that 
\begin{align}
\begin{split}
R(z)&=\frac14 z^2+\frac{-(c+1)z^2+p(c+2)z+(c^2-p^2)}{(z-p)^2}
\\
&=\frac{ ( (z-p)(z-2)-2c ) (  (z-p)(z+2)-2c  )   }{  4(z-p)^2 }=\frac{\prod_{j=1}^{4}(z-\lambda_j)}{4(z-p)^2},
\end{split}
\end{align}
where $\lambda_j$'s are given by \eqref{lambda 12} and \eqref{lambda 34}. 
Therefore by \eqref{R var}, the Stieltjes transform of $\mu_V$ is given by 
\begin{align}
\begin{split}
\int \frac{d\mu_V(s)}{z-s}&=\frac{V'(z)}{2}-R(z)^{1/2} =\frac{z}{2}-\frac{c}{z-p}-\frac12 \sqrt{ \frac{\prod_{j=1}^{4}(z-\lambda_j)}{(z-p)^2} }.
\end{split}
\end{align}
Letting $z=x+i\eps \to x \in \R$, we find
$$
\lim_{ \eps \to 0+ } \im \int \frac{d\mu_V(s)}{(x+i\eps)-s}=
\begin{cases}
\displaystyle \frac{\sqrt{ -\prod_{j=1}^{4}(x-\lambda_j) } }{2|x-p|} &\text{if } x\in [\lambda_1,\lambda_2]\cup [\lambda_3,\lambda_4],
\smallskip 
\\
0 &\text{otherwise}.
\end{cases}
$$
Now the desired identity \eqref{eq msr 1D} follows from the Sokhotski-Plemelj inversion formula, see e.g. \cite[Section \RN{1}.4.2]{MR1106850}.

\medskip

\subsection*{Acknowledgements} The author is greatly indebted to Yongwoo Lee for the figures and numerical simulations.

\bibliographystyle{abbrv}
\bibliography{RMTbib}

\end{document}